\documentclass[a4paper,12pt]{article}

\usepackage{setspace}
\usepackage{times} 
\usepackage{color}                      
\usepackage{float}

\usepackage[authoryear]{natbib}

\usepackage[top=1.25in, bottom=1.25in, left=1in, right=1in]{geometry}

\usepackage{rotating}

\usepackage{amsmath,amsthm,amsfonts,amssymb,bm} 

\newcommand\fullwidthdisplay{\displayindent0pt \displaywidth\columnwidth}
\AtBeginDocument{
  \everydisplay\expandafter{\expandafter\fullwidthdisplay\the\everydisplay}
}

\usepackage{makeidx}                            
\usepackage{listings}                          
\usepackage[colorlinks,
            linkcolor=black,
            anchorcolor=black,
            citecolor=black
            ]{hyperref}

\usepackage{graphicx,psfrag} 
\usepackage{epstopdf}
\usepackage{multirow}
\usepackage{booktabs}

\theoremstyle{plain}{

\newtheorem{proposition}{Proposition}[section]
\newtheorem{corollary}{Corollary}[section]
\newtheorem{lemma}{Lemma}[section]

}

\theoremstyle{definition}{
\newtheorem{defin}{Definition}[section]

\newtheorem{example}{Example}[section]

}

\numberwithin{equation}{section}
\begin{document}

\author{
Pengyu Wei\thanks{\rm Division of Banking and Finance, Nanyang Business School, Nanyang Technological University, Singapore. E-mail: pengyu.wei@ntu.edu.sg.} 
\and 
Wei Wei \thanks{Department of Actuarial Mathematics and Statistics, School of Mathematics and Computer Science, Heriot-Watt University, Edinburgh, Scotland, EH14, 4AS, UK.  E-mail: wei.wei@hw.ac.uk.}
}     
\title{Irreversible investment under weighted discounting: effects of decreasing impatience}  

\date{\today}
\maketitle

\thispagestyle{empty}

\bigskip

\centerline{\bf ABSTRACT}
This paper employs an intra-personal game-theoretic framework to investigate how decreasing impatience influences irreversible investment behaviors in a continuous-time setting. We consider a capacity expansion problem under weighted discount functions, a class of nonexponential functions that exhibit decreasing impatience, including the hyperbolic discount function as a special case. By deriving the Bellman system that characterizes the equilibrium, we establish the framework for analyzing investment behaviors of agents subject to decreasing impatience. From an economic perspective, we demonstrates that decreasing impatience prompts early investment. From a technical standpoint, we warn that decreasing impatience can lead to the failure of the smooth pasting principle.

\bigskip


\clearpage

\section{Introduction}
Many experimental and empirical studies have consistently identified a phenomenon known as decreasing impatience in intertemporal decision-making: individuals tend to discount future rewards at higher rates in the short term compared to the long term, rendering them ``present-biased" decision-makers. Consequently, they encounter time-inconsistent problems where optimal plans formulated today may not remain optimal when future dates arrive. Within the literature on time-inconsistent decision-making, such problems, often framed as self-control, are typically analyzed within the intra-personal game theoretic framework. The study of time inconsistency within this framework originates from seminal works \cite{strotz1955myopia} and \cite{phelps1968second}, then followed by a large literature including \cite{laibson1997golden}, \cite{o2001choice}, \cite{grenadier2007investment}, \cite{luttmer2003subjective}, and \cite{harris2013instantaneous}, among others.\par 
This paper contributes to the literature by analyzing an irreversible investment problem using instantaneous control in a continuous-time setting. While the time-consistent counterpart of this problem has been extensively explored within the real options literature to study investment behaviors in monopolistic, oligopolistic, and perfectly competitive markets \citep{dixit1994investment,grenadier2002option}, our focus lies in operationalizing general time preferences within this framework. Specifically, we investigate how decreasing impatience affects individual investment behaviors, deviating from the standard exponential discounting to consider discount functions exhibiting decreasing impatience, such as hyperbolic discount functions \citep{loewenstein1992anomalies,luttmer2003subjective}. This paper addresses two main issues: a) Methodology – establishing a Bellman system featuring decreasing impatience within the intra-personal game-theoretic framework and discussing the validity of the smooth pasting (SP) property typically used to construct explicit solutions for conventional instantaneous control problems; b) Economics – investigating the comparative statics of the investment trigger concerning the degree of decreasing impatience.\par 
In the methodology part, we highlight the potential failure of the SP principle and caution against its blind application to time-inconsistent problems. In standard time-consistent instantaneous control problems, the SP principle matches the marginal values on the investment trigger and yields the solution to the Bellman system and the corresponding optimal instantaneous control problem, whenever some mild conditions, such as the smoothness and convexity (or concavity) of the pay-off functions, are satisfied \citep{dumas1991super}. Recently, the SP principle has been extended to solve time-inconsistent stopping and instantaneous control problems. However, much of the literature takes the candidate solution obtained from the SP principle as the solution for the time-inconsistent problem without careful verification. We demonstrate that the SP principle solves the Bellman system and corresponding time-inconsistent problems if and only if a certain inequality holds, which may be violated for commonly used behavioral discount functions like the stochastic quasi-hyperbolic discount function\citep{grenadier2007investment,harris2013instantaneous}.
It is noteworthy that the failure of the SP principle was first identified by \cite{tan2021failure} in the context of solving time-inconsistent stopping problems. However, they did not uncover the underlying cause of this issue. In this paper, we explain the invalidity of the SP principle from a behavioral economics perspective, showing that decreasing impatience can lead to its failure.\par 
In the economics part, we focus on the impact of decreasing impatience on investment behaviors. According to \cite{prelec2004decreasing}, two factors determine the degree of impatience exhibited by a decision-maker: decreasing impatience and current time preference. The hyperbolic discount function offers a distinct advantage in studying decreasing impatience in stochastic models as both factors can be quantified by independent parameters. This characteristic positions the hyperbolic discount function as one of the most suitable candidates for analyzing decreasing impatience.\par 
A notable feature of the hyperbolic discount function is its demonstration of continuously changing time preferences, presenting a significant challenge in deriving the dynamics of the objective functional.\footnote{The continuous variation in time preferences implies a continuum of future selves. Consequently, if we use the recursive method in the classical literature on time-inconsistent control \citep[e.g.,][]{bjork2010general} to decompose the objective functional based on time preferences of all future selves, the resulting Bellman system would entail a pointwise combination of a continuum of differential equations. Particularly, each differential equation makes sense at only one point in time. Once the clock ticks, the differential equation representing the current self's dynamics will switch to the ``next" one. This contradicts the essence of the differential equation since the dynamics described by the equation can not evolve dynamically.} To address the challenge posed by continuously varying time preferences in stochastic models, we require further insights into discount functions. \cite{ebert2020weighted} observe that most discount functions, including hyperbolic discount functions, are completely monotone.\footnote{A function is completely monotone if the function and all its derivatives alternate in sign.} This property allows these functions to be represented as weighted averages of exponential functions \citep{bernstein1929fonctions}. This insight introduces a novel approach to deriving the Bellman system, departing from the conventional time-based decomposition. By decomposing the objective functional into expected discounted payoffs, with each component discounting future payoffs exponentially at distinct rates, we construct a set of differential equations that contribute to the Bellman system across the entire time horizon. In Section \ref{sec:solution}, we can see how this approach, involving the representation of discount functions in the weighted form, plays a crucial role in both explicitly solving the Bellman system and validating the SP principle.\par 
Decreasing impatience stands as a pivotal concept in understanding time preferences. Its quantitative measure, first established by \cite{prelec2004decreasing}, coupled with the explicit solutions derived in this paper, facilitates analytical evaluation of its impact on investment behaviors. We find that decreasing impatience leads to early investment. \cite{ebert2020weighted} also explore this issue through a stopping model within the real options framework but arrive at a starkly different conclusion. They find that decreasing impatience leads to delayed investment.\par 
The discrepancy in predictions between the two models within the real options approach stems from divergent interpretations of ``investment". In the optimal stopping model, investment entails obtaining immediate benefits by sacrificing future uncertainty. However, in the instantaneous control model, investment involves opting for higher-valued yet riskier projects by expanding capacity. Consequently, when viewed in the sense of risk attitude, the opposing predictions converge to a consistent interpretation: decreasing impatience prompts risk-taking investment.\par 
The remainder of the paper is organized as follows. Section \ref{sec:model} formulates the irreversible investment problem within the intra-personal game theoretic framework and characterizes the equilibrium by a Bellman system. In Section \ref{sec:solution}, we solve the Bellman system by the SP principle and formally establish the inequality without which the SP solution is not an equilibrium. Section \ref{sec:decreasing impatience} explores the impact of decreasing impatience on investment behaviors and the effectiveness of the SP principle. Section \ref{sec:conclusion} concludes. All proofs are collected in the appendix.

\section{The Model}\label{sec:model}
\subsection{Time preferences}
In this paper, we consider time preferences modeled by weighted discount functions, defined as follows.

\begin{defin}[\cite{ebert2020weighted}]
Let $h:[0, \infty) \rightarrow(0,1]$ be decreasing with $h(0)=1$. The function $h$ is a weighted discount function if there exists a distribution function $F(r)$ concentrated on $[0, \infty)$ such that\footnote{Following the conventions in mathematical analysis, we mean the integration area is $(a, b]$ if we write $\int_{a}^{b}, \forall a<$ $b<\infty$. We let $\int_{a}^{\infty}$ define $\lim _{b \rightarrow \infty} \int_{a}^{b}$.} 
\begin{equation}\label{eq:weighted h}
    h(t)=\int_{0}^{\infty} e^{-r t} d F(r),
\end{equation}
and $F$ is called the weighting distribution of $h$.
\end{defin}

\cite{ebert2020weighted} find that most commonly used discount functions can be represented as weighted discount functions. Specifically, the hyperbolic discount function of \cite{luttmer2003subjective}, the most well-documented non-exponential discount function, can be written in a weighted form\footnote{In behavioral science, the linkage between hyperbolic discount function and the Gamma distribution is first found by \cite{sozou1998hyperbolic}, who shows that if individuals are uncertain about what discount rates to use, then the discount function would be hyperbolic.}
\begin{equation}\label{eq:hyperbolic h}
    h(t)=\frac{1}{(1+\alpha t)^{\frac{\beta}{\alpha}}} e^{-\phi t}=\int_{\phi}^{\infty} e^{-r t} f\left(r-\phi ; \frac{\beta}{\alpha}, \alpha\right) d r,
\end{equation}
where
$$
f(r ; k, \theta)=\frac{r^{k-1} e^{-\frac{r}{\theta}}}{\theta^{k} \Gamma(k)}, \Gamma(k)=\int_{0}^{\infty} x^{k-1} e^{-x} d x.
$$

To our knowledge, \eqref{eq:hyperbolic h} epresents the most general form of the hyperbolic discount function. When $\phi=0$, \eqref{eq:hyperbolic h} reduces to $1/(1+\alpha t)^{\frac{\beta}{\alpha}}$, which is the hyperbolic discount function introduced by \cite{loewenstein1992anomalies}. Furthermore, by setting $\alpha=\beta$, we obtain the hyperbolic discount function $1/(1+\alpha t)$, which is the discount function studied by \cite{harvey1995proportional}.

In addition, exponential discount function \citep{samuelson1937note}, pseudo-exponential discount function \citep{karp2007non,ekeland2008investment}, constant sensitivity (CS) discount function \citep{ebert2007fragility} and constant absolute decreasing impatience (CADI) discount function \citep{bleichrodt2009non} are examples of weighed discount functions.

\subsection{Economic setup}
Consider a monopolistic industry in which the firm chooses how many units of output to produce. The output is infinitely divisible and each unit costs $K$. Suppose that the firm produces $q$ units of output at time $t$ and the price per unit of output, $P_{t}$, is given by $P_{t}=D\left(X_{t}, q\right)$, where $D$ is the inverse demand function and $X_{t}$ is a multiplicative shock process which follows the following stochastic differential equation (SDE)
\begin{equation}\label{eq:X_t}
    d X_{t}=\mu\left(X_{t}\right) d t+\sigma\left(X_{t}\right) d W_{t}
\end{equation}
with $W=\{W\}_{t \geq 0}$ being a standard Brownian motion.

Moreover, to guarantee SDE \eqref{eq:X_t} has a unique strong solution, we, following standard literature \citep[e.g.,][]{karatzas2014brownian}, suppose that the functions $\mu$ and $\sigma$ are Lipschitz continuous, i.e., there exists $L>0$, such that for all $x_{1} \neq x_{2}$,
\begin{equation}\label{eq:Lipschitz}
    \left|\mu\left(x_{1}\right)-\mu\left(x_{2}\right)\right|+\left|\sigma\left(x_{1}\right)-\sigma\left(x_{2}\right)\right|<L\left|x_{1}-x_{2}\right|
\end{equation}

Under condition \eqref{eq:Lipschitz}, SDE \eqref{eq:X_t} admits a unique strong solution such that for any $T>0, m>0,$
\begin{equation}\label{eq:C_T}
    \sup _{0 \leq t \leq T}\left|X_{t}\right|^{m} \leq C_{T}\left(1+|x|^{m}\right),
\end{equation}
where $C_T$ is a constant \citep[see Chapter 1 of][]{yong2012stochastic}.

The firm faces two options at any given time: maintain the current output level $q$ until the next decision point or expand its capacity to increase output. We define the investment strategy by $u^{q}$, a function of time $t$ and the shock process value $X$. We suppose that the investment is irreversible, and thus $u^{q}$ takes values in $[q, \infty)$. The set of functions $U=\left\{u^{q}\right\}_{q>0}$, for a given sample path of $X$, gives the firm's output at each calendar date $t$ and thus defines the paths of the output. Moreover, as $u^{q}$ takes value in $[q, \infty)$, the paths of output are increasing. We define the output paths determined by the investment strategy set $U$ by $Q^{U}=\left\{Q_{t}^{U}\right\}_{t \geq 0}$.

For simplicity, we assume there is no variable costs, and thus the profit flow at time $t$ is given by $P_{t}$. We assume that the discount function can be written in weighted form and the objective functional of the firm is given by the expected discounted future cash flows\footnote{To be mathematically rigorous, $Q_{t}^{U}$ should satisfy some regularity conditions if we write it as an integrator. Here we follow standard conventions in stochastic calculus, supposing that $Q^{U}$ is increasing and right continuous with left limits. The increase of $Q^{U}$ is guaranteed by the value set of $u^{q} \in U$ while the right continuity and left limits show that $Q^{U}$ is the output path after instantaneous capacity expansion.}

\begin{equation}\label{eq:objective}
    \begin{aligned}
        & \mathbb{E}\left[\int_{t}^{\infty} h(s-t) P_{s} Q_{s}^{U} d s-\int_{t}^{\infty} K h(s-t) d Q_{s}^{U} \mid X_{t}=x\right] \\
        = & \mathbb{E}\left[\int_{t}^{\infty} h(s-t) \Pi\left(X_{s}, Q_{s}^{U}\right) d s-\int_{t}^{\infty} K h(s-t) d Q_{s}^{U} \mid X_{t}=x\right],
    \end{aligned}
\end{equation}
where $\Pi(x, q)=D(x, q) q$. We assume $\Pi$ and $\frac{\partial \Pi}{\partial q}$ have polynomial growth with respect to $x$, i.e., there exist $m>0, C(q)>0$ such that 
$$|\Pi(x, q)|+\left|\frac{\partial \Pi}{\partial q}(x, q)\right|<C(q)\left(|x|^{m}+1\right).$$

\subsection{The standard time-consistent real options case}
We provide a concise overview of a time-consistent irreversible investment problem as found in standard real options literature \citep[e.g., Chapter 11 of][]{dixit1994investment}. Suppose 
$$h(t)=e^{-r_{0} t}, r_{0}>0,$$ 
and consider the following optimization problem with current output $q$
\begin{equation}\label{eq:time-consistent}
    \sup _{U} \mathbb{E}\left[\int_{t}^{\infty} e^{-r_{0}(s-t)} X_{s}\left(Q_{s}^{U}\right)^{1-\frac{1}{\gamma}} d s-\int_{t}^{\infty} K e^{-r_{0}(s-t)} d Q_{s}^{U} \mid X_{t}=x\right],
\end{equation}
where we assume that the inverse demand function takes a constant-elasticity form, i.e., the price $P_{t}$ is
given by 
$P_{t}=X_{t} q^{-\frac{1}{\gamma}}$, and the multiplicative shock process $X=\left\{X_{t}\right\}_{t \geq 0}$ follows a geometric Brownian motion 
$d X_{t}=\sigma X_{t} d W_{t}$ with $\sigma>0$.

Define the optimal value by $V^{o}(x, q)$ and the optimal investment threshold by $x^{o}(q)$. Then the SP principle yields that $\left(V^{o}(x, q), x^{*}(q)\right)$ solves the following differential equation on $\left(0, x^{o}(q)\right)$
$$
\frac{1}{2} \sigma^{2} x^{2} \frac{\partial^{2} V^{o}}{\partial x^{2}}+q^{1-\frac{1}{\gamma}} x-r_{0} V^{o}=0,
$$
with the boundary conditions
$$
\begin{gathered}
\frac{\partial V^{o}}{\partial q}(0, q)=0, \\
\frac{\partial V^{o}}{\partial q}\left(x^{o}(q), q\right)=K, \\
\frac{\partial^{2} V^{o}}{\partial q \partial x}\left(x^{o}(q), q\right)=0 .
\end{gathered}
$$

Solving the above differential equation, we have that the optimal investment threshold is given by
\begin{equation}\label{eq:xoq}
    x^{o}(q)=\frac{\theta\left(r_{0}\right)}{\theta\left(r_{0}\right)-1} r_{0} \frac{\gamma}{\gamma-1} q^{\frac{1}{\gamma}} K,
\end{equation}
where $\theta$ is the positive square root of $\frac{1}{2} \sigma^{2} \theta^{2}-\frac{1}{2} \sigma^{2} \theta-r=0$, i.e.,
\begin{equation*}
    \theta(r)=\frac{\frac{1}{2} \sigma^{2}+\sqrt{\frac{1}{4} \sigma^{4}+2 \sigma^{2} r}}{\sigma^{2}}
\end{equation*}

Consequently, the optimal investment strategy $u_{q}^{o}$ is given by

$$
u_{q}^{o}(x)=\left\{\begin{aligned}
q, & \quad \text {if } x<x^{o}(q), \\
\inf \left\{q: \frac{\partial V}{\partial q}(x, q)=K\right\}, & \quad \text {otherwise.}
\end{aligned}\right.
$$

Computing the expectation in \eqref{eq:time-consistent}, we obtain the optimal value function $V^{o}$ given by\footnote{See \cite{grenadier2002option} for details.}\textsuperscript{,}\footnote{To guarantee that the solution of the investment problem is finite, we, following standard literature \citep[][]{dixit1994investment,grenadier2002option}, suppose that $1<\gamma<\theta\left(r_{0}\right)$ throughout this paper.}
\begin{equation*}
    \begin{aligned}
        &V^{o}(x, q)\\
        =&\left\{\begin{array}{ll}
K\left(1-\frac{\iota^{o}}{r_{0}}\right)\left(\frac{\gamma-1}{\gamma \iota^{\circ} K}\right)^{\theta\left(r_{0}\right)} \frac{\gamma}{\gamma-\theta\left(r_{0}\right)} q^{1-\frac{\theta\left(r_{0}\right)}{\gamma}}+\frac{x}{r_{0}} q^{1-\frac{1}{\gamma}}, &x \leq x^{o}(q) \\
K\left(1-\frac{\iota}{r_{0}}\right) x^{\gamma} \frac{\gamma}{\gamma-\theta\left(r_{0}\right)}\left(\frac{\gamma-1}{\gamma \iota^{\circ} K}\right)^{\gamma}+\frac{x^{\gamma}}{r_{0}}\left(\frac{\gamma-1}{\gamma \iota^{\circ} K}\right)^{\gamma-1}-K\left(\left(\frac{x(\gamma-1)}{\gamma \iota^{\circ} K}\right)^{\gamma}-q\right),&  x>x^{o}(q),
\end{array}\right.
    \end{aligned}
\end{equation*}
where $\iota^{o}=r_{0} \frac{\theta\left(r_{0}\right)}{\theta\left(r_{0}\right)-1}$.

\subsection{Equilibrium}
In an intra-personal game, a decision is made by a sequence of selves of the decision maker. Self $t$, representing the decision maker standing at time $t$, controls her investment strategy at time $t$ and takes into account her future plans. Each self aims to optimize her current objective under the assumption that the strategy made by the future selves will be followed. Following the literature on time-inconsistent decision-making \citep[e.g.,][]{grenadier2007investment,harris2013instantaneous}, we pursue a stationary Markov equilibrium. The stationarity implies that all future selves will use the same strategy and solve the same problem, and thus we only need to consider self 0's problem. In more detail, we define
\begin{equation}\label{eq:objective weighted}
J^{q}(x ; U):  =\mathbb{E}\left[\int_{0}^{\infty} h(s) \Pi\left(X_{s}, Q_{s}^{U}\right) d s-\int_{0}^{\infty} K h(s) d Q_{s}^{U} \mid X_{0}=x\right],
\end{equation}
where $q$ marks the initial output and $x$ gives the starting value of the shock process. We now formally define an equilibrium investment policy.\par 
\begin{defin}
    The set of investment strategies $\bar{U}=\left\{\bar{u}^{q}(x)\right\}_{q>0}$ is an equilibrium investment policy if for $\forall q>0, x>0$,
\begin{equation}\label{eq:equilibrium}
    \liminf _{\epsilon \rightarrow 0} \frac{J^{q}(x ; \bar{U})-\left(J^{q}\left(x ; \bar{U}^{\epsilon, a}\right)-K(a-q)\right)}{\epsilon} \geq 0
\end{equation}
where $\bar{U}^{\epsilon, a}=\left\{u^{q}\right\}_{q>0}$ is defined as
\begin{equation}\label{eq:U bar}
    u^{q}(s, x)=\left\{\begin{aligned}
\bar{u}^{q}(x) & \text { if } s \in[\epsilon, \infty), \\
a & \text { if } s \in[0, \epsilon),
\end{aligned}\right.
\end{equation}
with $a \in[q, \infty), \epsilon>0, x>0.$
\end{defin}\par 
Our equilibrium definition is consistent with those from time-inconsistent control problems \cite[e.g.,][]{bjork2010general,ekeland2012time,bjork2017time} when interpreting the equilibrium investment problem as a control problem\footnote{There are a plenty of discussions about intra-personal equilibria. See, for example, \cite{christensen2020time}, \cite{bayraktar2021equilibrium}, \cite{he2021equilibrium}, \cite{huang2021strong} and \cite{bodnariu2024local}.}. The perturbed investment policy $U^{\epsilon, a}$ coincides with the equilibrium policy except for a short time interval in which $U^{\epsilon, a}$ explores all  expansion choices. Therefore $U^{\epsilon, a}$ constitutes all possible (local) deviation rules. Compared with time-inconsistent control problems considered in the literature \citep[e.g.,][]{bjork2010general}, our definition involves an extra term $-K(a-q)$, which represents the instantaneous cost due to the immediate increase of output from $q$ to $a$. It follows from the inequality \eqref{eq:equilibrium} that for a firm staying in an intra-personal equilibrium, any local deviation is undesirable.\par 
\subsection{Equilibrium characterization}
In this section, we characterize the equilibrium investment policy by a Bellman system.\par 
\begin{proposition}[Equilibrium characterization]\label{prop:equilibrium}
    Consider the objective functional \eqref{eq:objective} with weighted discount function $h(t)=\int_{0}^{\infty} e^{-r t} d F(r)$, an investment policy $\bar{U}=\left\{\bar{u}^{q}(x)\right\}_{q>0}$, underlying process $X$ defined by \eqref{eq:X_t} and functions
\begin{equation}\label{eq:w}
    w(x, q ; r)=\mathbb{E}\left[\int_{0}^{\infty} e^{-r t} \Pi\left(X_{t}, Q_{t}^{\bar{U}}\right) d t-\int_{0}^{\infty} e^{-r t} K d Q_{t}^{\bar{U}} \mid X_{0}=x\right]
\end{equation}

and

$$
V(x, q)=\int_{0}^{\infty} w(x, q ; r) d F(r)
$$

Suppose that 
\begin{enumerate}
    \item $w$ and $V$ are continuously differentiable in $q$;

    \item $\frac{\partial w}{\partial q}$ and $w$ have polynomial growth in $x$, i.e., there exist $C(r, q)>0, m>0$, such that $$|w(x, q ; r)|+\left|\frac{\partial w}{\partial q}(x, q ; r)\right| \leq C(r, q)\left(|x|^{m}+1\right),$$ where $\int_{0}^{\infty} r C(r, q) d F(r)+\int_{0}^{\infty} C(r, q) d F(r)<\infty$;

    \item $w$ and $\frac{\partial w}{\partial q}$ are continuous in $x$;

    \item $V$ and $\frac{\partial V}{\partial q}$ are continuously differentiable in $x$ and the corresponding first-order derivatives are locally Lipschitz continuous.
\end{enumerate}
Moreover, we suppose that the triplet $(V, w, \bar{U})$ solves
\begin{align}
& \max \left\{\frac{1}{2} \sigma^{2}(x) \frac{\partial^{2} V}{\partial x^{2}}+\mu(x) \frac{\partial V}{\partial x}+\Pi(x, q)-\int_{0}^{\infty} r w(x, q ; r) d F(r), \frac{\partial V}{\partial q}-K\right\}=0, \label{eq:HJB}\\
& \bar{u}^{q}(x)=\left\{\begin{array}{cc}
q, & \quad \text{if } \frac{\partial V}{\partial q}(x, q)<K, \\
\inf \left\{q: \frac{\partial V}{\partial q}(x, q)=K\right\}, & \quad \text{otherwise, }\notag
\end{array}\right.
\end{align}
and for $x>0$, $q_{1}>q_{2}>0$, $\bar{u}$ satisfies the following condition
\begin{equation}\label{eq:condition uq}
    \bar{u}^{q_{1}}(x)=q_{1}, \text { if } \bar{u}^{q_{2}}(x)=q_{2}
\end{equation}

Then, $\bar{U}$ is an equilibrium investment policy and the value function is given by $V(x, q)$, i.e., $V(x, q)=J^{q}(x ; \bar{U})$
\end{proposition}

We provide an intuitive explanation for Proposition \ref{prop:equilibrium}. A Bellman system captures the local optimum for a dynamic problem.\footnote{If the weighting distribution is degenerate, then $V(x, q)=w(x, q ; r)$ and \eqref{eq:HJB} becomes the well-known Bellman equation $\max \left\{\frac{1}{2} \sigma^{2}(x) \frac{\partial^{2} V}{\partial x^{2}}+\mu(x) \frac{\partial V}{\partial x}-r V+\Pi(x, q), \frac{\partial V}{\partial q}-K\right\}=0$, which is used to solve conventional time-consistent investment problems.} This property coincides with the essence of intra-personal equilibrium and hence making Bellman systems prevalent in solving time-inconsistent problems. Intuitively, a Bellman equation compares a class of dynamics of the objective functional and selects the best one as the evolution of the optimal value function. In our investment problem, consider two scenarios: when the firm increases its output from $q$ to $q^{\prime}$, and when it maintains its current output in an infinitesimal period $dt$. In the former case, the value function $V(x,q)$ is affected by the immediate cost of the output increase
$$V(x, q)= V\left(x, q^{\prime}\right)-K\left(q^{\prime}-q\right).$$ 
In the latter case, the value function evolves based on the expected future value, adjusted by the running profit
$$V(x, q)=\mathbb{E}\left[V\left(x+dX_t, q\right) \mid X_{0}=x\right]+ \text{Running Profit}  \times d t.$$ 
Therefore, the (local) optimal value of the firm is determined by the maximum of these two scenarios.
$$V(x, q)=\max \left\{\mathbb{E}\left[V\left(x+dX_t, q\right) \mid X_{0}=x\right]- \text{Running Profit}  \times d t, V\left(x, q^{\prime}\right)-K\left(q^{\prime}-q\right)\right\}.$$
By subtracting $V(x, q)$ from both sides and applying Ito's formula, we have 
$$0=\max \left\{ \text{Dynamics of }V, \frac{\partial V}{\partial q}- K \right\}.$$

Therefore, a crucial step in obtaining a Bellman system is to derive the dynamics of the objective functional. While most literature on time-inconsistent decision-making addresses this issue by decomposing the objective functional based on the time preferences of all selves of the decision maker \citep[e.g.,][]{grenadier2007investment,harris2013instantaneous}, we decompose the objective functional into a set of expected discounted payoffs, where the time preferences are modeled by exponential functions. Our decomposition relies on the availability of the weighted representation for the discount function. It is noteworthy that \cite{ebert2020weighted} observe that most commonly used discount functions can be written in a weighted form, making our approach applicable to general discount functions. Through this decomposition, we obtain a new representation of the dynamics of the value function. As shown in Proposition \ref{prop:equilibrium}, we have 
$$V(x, q)=\int_{0}^{\infty} w(x, q ; r) d F(r),$$ and the dynamics of $w(x, q ; r)$ in the continuation region is given by the differential equation $$\frac{1}{2} \sigma^{2}(x) \frac{\partial^{2} w}{\partial x^{2}}+\mu(x) \frac{\partial w}{\partial x} + \Pi(x, q) -r w =0.$$ 
Consequently, the differential equation the value function $V$ should satisfy in the continuation region is given by $$\frac{1}{2} \sigma^{2}(x) \frac{\partial^{2} V}{\partial x^{2}}+\mu(x) \frac{\partial V}{\partial x} +\Pi(x, q) -\int_{0}^{\infty} r w(x, q ; r) d F(r)=0.$$

We conclude this section by explaining condition \eqref{eq:condition uq}, which is to exclude some unreasonable equilibria. Essentially, condition \eqref{eq:condition uq} stipulates that in an equilibrium, given a fixed value of the shock process, if the firm opts not to increase its current output, it should not do so for a higher level of output either. This condition thus precludes the following nonsensical scenario: where a lower level of output suffices for the firm to attain optimality, yet a higher level does not.

\section{Solutions}\label{sec:solution}
The section presents an explicit solution to the Bellman system in Proposition \ref{prop:equilibrium} through a specific example, whose time-consistent counterpart is commonplace in standard real options literature. To ensure the explicit solution obtained in the following is finite and the interchange of integral and derivative is valid, we assume
$$\max \left\{\int_{0}^{\infty} \frac{1}{r} d F(r), \int_{0}^{\infty} r d F(r)\right\}<\infty .$$
This condition is satisfied by the hyperbolic discount function \eqref{eq:hyperbolic h} whenever $\phi>0$.

We consider an inverse demand function of constant-elasticity form
$$P_{t}=X_{t} q^{-\frac{1}{\gamma}}$$ 
with $\gamma>1$. Following standard literature \citep[e.g.,][]{abel1983optimal}, we further assume that the multiplicative shock process $X=\left\{X_{t}\right\}_{t \geq 0}$ follows a geometric Brownian motion 
$$
d X_{t}=\sigma X_{t} d W_{t},
$$
where $\sigma$ is a positive constant.

\subsection{The smooth-pasting principle}
We solve the Bellman system in Proposition \ref{prop:equilibrium} by the SP principle, which states that the marginal profit and its derivative with respect to the value of the shock process are equal at the optimal threshold.

For $\forall q>0$, we define the boundary that demonstrates the continuation region and expansion region by $x^{*}(q)$. It follows from \eqref{eq:w} that $w$ in the expansion region $\left[x^{*}(q), \infty\right)$ satisfies
$$
\frac{\partial w}{\partial q}(x, q ; r)=K
$$

Note that $\frac{\partial w}{\partial q}$ solves the following differential equation in the continuation region $\left(0, x^{*}(q)\right)$,
$$
\frac{1}{2} \sigma^{2} x^{2} \frac{\partial^{2} \frac{\partial w}{\partial q}}{\partial x^{2}}(x, q ; r)-r \frac{\partial w}{\partial q}(x, q ; r)+\left(1-\frac{1}{\gamma}\right) q^{-\frac{1}{\gamma}}=0
$$
then we have that
$$
\frac{\partial w}{\partial q}(x, q ; r)=\left(K-\frac{x^{*}(q)}{r}\left(1-\frac{1}{\gamma}\right) q^{-\frac{1}{\gamma}}\right)\left(\frac{x}{x^{*}(q)}\right)^{\theta(r)}+\frac{x}{r}\left(1-\frac{1}{\gamma}\right) q^{-\frac{1}{\gamma}},
$$
where $\theta$ is the positive square root of $\frac{1}{2} \sigma^{2} \theta^{2}-\frac{1}{2} \sigma^{2} \theta-r=0$, i.e.,
\begin{equation}\label{eq:theta r}
    \theta(r)=\frac{\frac{1}{2} \sigma^{2}+\sqrt{\frac{1}{4} \sigma^{4}+2 \sigma^{2} r}}{\sigma^{2}}
\end{equation}

Since $V$ is the weighted average of $w$, then we have
$$
\begin{aligned}
&\frac{\partial V}{\partial q}(x, q) \\ =&\int_{0}^{\infty} \frac{\partial w}{\partial q}(x, q ; r) d F(r) \\
=&\int_{0}^{\infty}\left(K-\frac{x^{*}(q)}{r}\left(1-\frac{1}{\gamma}\right) q^{-\frac{1}{\gamma}}\right)\left(\frac{x}{x^{*}(q)}\right)^{\theta(r)} d F(r)+\int_{0}^{\infty} \frac{x}{r}\left(1-\frac{1}{\gamma}\right) q^{-\frac{1}{\gamma}} d F(r)
\end{aligned}
$$

Therefore, the SP principle $\frac{\partial^{2} V}{\partial q \partial x}\left(x^{*}(q), q\right)=0$ yields that

$$
\int_{0}^{\infty}\left(K-\left(1-\frac{1}{\gamma}\right) q^{-\frac{1}{\gamma}} \frac{x^{*}(q)}{r}\right) \theta(r) \frac{1}{x^{*}(q)} d F(r)+\left(1-\frac{1}{\gamma}\right) q^{-\frac{1}{\gamma}} \int_{0}^{\infty} \frac{1}{r} d F(r)=0
$$

Then we have
\begin{equation}\label{eq:x*}
    x^{*}(q)=\frac{\int_{0}^{\infty} \theta(r) d F(r)}{\int_{0}^{\infty} \frac{\theta(r)-1}{r} d F(r)} \frac{\gamma}{\gamma-1} q^{\frac{1}{\gamma}} K.
\end{equation}

Obviously, if the weighting distribution $F$ is degenerate at $r_{0}$, i.e., $h(t)=e^{-r_{0} t}$, then the boundary is identical to the one in \cite{grenadier2002option}, i.e., \eqref{eq:xoq}.

\begin{example}
   For the exponential discount function $h(t)=e^{-r_{0} t}, r_{0}>0$,
\begin{equation}\label{eq:x* exponential}
    x^{*}(q)=\frac{\theta\left(r_{0}\right)}{\theta\left(r_{0}\right)-1} r_{0} \frac{\gamma}{\gamma-1} q^{\frac{1}{\gamma}} K.
\end{equation}
\end{example}

We present the irreversible investment under hyperbolic discounting in the following example.
\begin{example}
   For the hyperbolic discount function \eqref{eq:hyperbolic h},
\begin{equation}\label{eq:x* hyperbolic}
    x^{*}(q)=\frac{\int_{\phi}^{\infty} \theta(r) f\left(r-\phi ; \frac{\beta}{\alpha}, \alpha\right) d r}{\int_{\phi}^{\infty} \frac{\theta(r)-1}{r} f\left(r-\phi ; \frac{\beta}{\alpha}, \alpha\right) d r} \frac{\gamma}{\gamma-1} q^{\frac{1}{\gamma}} K,
\end{equation}
where
$$
f(r ; k, \theta)=\frac{r^{k-1} e^{-\frac{r}{\theta}}}{\theta^{k} \Gamma(k)}, \Gamma(k)=\int_{0}^{\infty} x^{k-1} e^{-x} dx.
$$
\end{example}

Given $x^{*}(q)$, we can obtain the value function $V$ and its component $w$ by the standard argument in the irreversible investment literature \citep[e.g.,][]{dixit1994investment,grenadier2002option}. For $x \leq x^{*}(q)$, we have\footnote{In order to ensure $w(x, q ; r)$ is finite, we suppose that $\gamma<\theta(r)$, see, for example, footnote 17 of \cite{grenadier2002option}.}
$$
w(x, q ; r)=C(q) x^{\theta}+\frac{x}{r} q^{1-\frac{1}{\gamma}}, x \leq x^{*}(q),
$$
where
$$
\begin{aligned}
    C(q)=&-\int_{q}^{\infty}\left(K-\frac{x^{*}(s)}{r}\left(1-\frac{1}{\gamma}\right) s^{-\frac{1}{\gamma}}\right)\left(\frac{1}{x^{*}(s)}\right)^{\theta(r)} d s\\
    =&K\left(1-\frac{\iota}{r}\right)\left(\frac{\gamma-1}{\gamma \iota K}\right)^{\theta(r)} \frac{\gamma}{\gamma-\theta(r)} q^{1-\frac{\theta(r)}{\gamma}}
\end{aligned}
$$
and
$$
\iota=\frac{\int_{0}^{\infty} \theta(r) d F(r)}{\int_{0}^{\infty} \frac{\theta(r)-1}{r} d F(r)}
$$

For $x>x^{*}(q)$, the firm increases its output with minimum effort such that $x$ falls in the continuation region. Suppose that $\tilde{q}$ solves the equation $x=x^{*}(q)$, i.e., $$\tilde{q}=\left(\frac{x(\gamma-1)}{\gamma \kappa K}\right)^{\gamma}.$$ 
Then
$$
\begin{aligned}
&w(x, q ; r) \\
=&w(x, \tilde{q} ; r)-K(\tilde{q}-q) \\
=&K\left(1-\frac{\iota}{r}\right) x^{\gamma} \frac{\gamma}{\gamma-\theta(r)}\left(\frac{\gamma-1}{\gamma \iota K}\right)^{\gamma}+\frac{x^{\gamma}}{r}\left(\frac{\gamma-1}{\gamma \iota K}\right)^{\gamma-1}-K\left(\left(\frac{x(\gamma-1)}{\gamma \iota K}\right)^{\gamma}-q\right), x>x^{*}(q)
\end{aligned}
$$

Finally, it follows from $V(x, q)=\int_{0}^{\infty} w(x, q ; r) d F(r)$ that
$$
\begin{aligned}
    &V(x, q)\\
    =&\left\{\begin{array}{c}
\int_{0}^{\infty} K\left(1-\frac{\iota}{r}\right)\left(\frac{\gamma-1}{\gamma \iota K}\right)^{\theta(r)} \frac{\gamma}{\gamma-\theta(r)} q^{1-\frac{\theta(r)}{\gamma}} d F(r)+\int_{0}^{\infty} \frac{x}{r} q^{1-\frac{1}{\gamma}} d F(r), x \leq x^{*}(q) \\
\int_{0}^{\infty}\left(K\left(1-\frac{\iota}{r}\right) x^{\gamma} \frac{\gamma}{\gamma-\theta(r)}\left(\frac{\gamma-1}{\gamma \iota K}\right)^{\gamma}+\frac{x^{\gamma}}{r}\left(\frac{\gamma-1}{\gamma \iota K}\right)^{\gamma-1}\right) d F(r)-K\left(\left(\frac{x(\gamma-1)}{\gamma \iota K}\right)^{\gamma}-q\right), x>x^{*}(q).
\end{array}\right.
\end{aligned}
$$

\subsection{Verification of the SP principle}
 In the classical time-consistent irreversible investment problem, the SP principle always yields the solution to the Bellman system with some mild conditions on the coefficients. However, this is not the case when the time consistency is lost. The following proposition illustrates that the candidate solution obtained by the SP principle solves the Bellman system if and only if a certain inequality is satisfied.

\begin{proposition}\label{prop:verification}
    With the notations above, the triplet $(V(x, q), w(x, q ; r), \bar{U}(x))$ solves the Bellman system in Proposition \ref{prop:equilibrium} if and only if
\begin{equation}\label{eq:condition inequality}
    \int_{0}^{\infty} \theta(r) d F(r) \geq \int_{0}^{\infty} r d F(r) \int_{0}^{\infty} \frac{\theta(r)-1}{r} d F(r).
\end{equation}
\end{proposition}\par 
Inequality \eqref{eq:condition inequality} serves as a crucial condition in the construction of the explicit solution to the time-inconsistent problem. It underscores the necessity of verifying the candidate solution obtained through the SP principle. Notably, when the distribution function $F$ is degenerate, inequality \eqref{eq:condition inequality} is automatically satisfied. In such cases, our solution aligns with the solution presented in \cite{grenadier2002option}.\par 
However, the subsequent corollary demonstrates that inequality \eqref{eq:condition inequality} may not hold, even for the simplest non-exponential discount function. This highlights the importance of careful validation when employing the SP principle to derive solutions in time-inconsistent settings.\par 
\begin{corollary}\label{coro: stochastic quasi}
    Suppose the weighting distribution function is
    \begin{equation}\label{eq:weighting pseudo}
        F(s)=\left\{\begin{array}{cc}0 & s<r, \\ \delta & r \leq s<r+\lambda, \\ 1 & \text { otherwise, }\end{array}\right.
    \end{equation}
    with $r,\lambda \ge 0 $ and $0<\delta<1$. Then
    \begin{enumerate}
        \item there exists $\lambda_{1}>0$ such that condition \eqref{eq:condition inequality} holds whenever $\lambda \in\left(0, \lambda_{1}\right)$. In this case, the equilibrium investment triggering boundary is given by
        \begin{equation}\label{eq:barrier pseudo}
            x^{*}(q)=\frac{\delta \theta(r)+(1-\delta) \theta(r+\lambda)}{\delta \frac{\theta(r)-1}{r}+(1-\delta) \frac{\theta(r+\lambda)-1}{r+\lambda}} \frac{\gamma}{\gamma-1} q^{\frac{1}{\gamma}} K;
        \end{equation}

        \item there exists $\lambda_{2}>0$ such that condition \eqref{eq:condition inequality} does not hold whenever $\lambda \in\left(\lambda_{2}, \infty\right)$.
    \end{enumerate}
\end{corollary}\par 
The weighting distribution \eqref{eq:weighting pseudo} corresponds to the pseudo-exponential discount function \citep{karp2007non,ekeland2008investment}
\begin{equation}\label{eq:pseudo}
    h(t)=\delta e^{-r t}+(1-\delta) e^{-(r+\lambda) t},
\end{equation}
which is closely related to the stochastic quasi-hyperbolic discount function \citep{grenadier2007investment}, i.e., at time $t$ the agent's self $n$ applies the discount factor $D_{n}(t, s)$ to a future payoff at time $s$ given by
\begin{equation}\label{eq:stochastic quasi}
    D_{n}(t, s)=\left\{\begin{array}{rr}
e^{-r(s-t)}, & \text { if } s \in\left[t_{n}, t_{n+1}\right), \\
\delta e^{-r(s-t)}, & \text { if } s \in\left[t_{n+1}, \infty\right),
\end{array}\right.
\end{equation}
where $\left\{t_{n}\right\}_{n \geq 1}$ is a sequence of arrival times that follow a Poisson process with intensity $\lambda$ and independent of the shock process $X$. As shown in \cite{harris2013instantaneous}, the pseudo-exponential and stochastic quasi-hyperbolic discount functions have an equivalent effect on decision-making, as the preference of an agent is given by the expectation of the discounted payoffs.\footnote{To illustrate this equivalence for our model, it suffices to show the two discount functions yield the same objective functional. In fact, it is easy to see that $\mathbb{E}\left[D_{n}(t)\right]=\delta e^{-r(s-t)}+(1-\delta) e^{-(r+\lambda)(s-t)}$. Therefore, from the independence between the discount function and shock process, we have that the objective functional \eqref{eq:objective} is given by
$$
\begin{aligned}
& \mathbb{E}\left[\int_{t}^{\infty} D_{n}(s-t) X_{s}\left(Q_{s}^{U}\right)^{1-\frac{1}{\gamma}} d s-\int_{t}^{\infty} K D_{n}(s-t) d Q_{s}^{U} \mid X_{t}=x\right] \\
= & \mathbb{E}\left[\int_{t}^{\infty} \mathbb{E}\left[D_{n}(s-t)\right] X_{s}\left(Q_{s}^{U}\right)^{1-\frac{1}{\gamma}} d s-\int_{t}^{\infty} K \mathbb{E}\left[D_{n}(s-t)\right] d Q_{s}^{U} \mid X_{t}=x\right] \\
= & \mathbb{E}\left[\int_{t}^{\infty} h(s-t) X_{s}\left(Q_{s}^{U}\right)^{1-\frac{1}{\gamma}} d s-\int_{t}^{\infty} K h(s-t) d Q_{s}^{U} \mid X_{t}=x\right],
\end{aligned}
$$
where $h(t)=\delta e^{-r t}+(1-\delta) e^{-(r+\lambda) t}$.} Given the prevalence of the stochastic quasi-hyperbolic discount function in dynamic decision-making, the equivalence and Corollary \ref{coro: stochastic quasi} caution that the violation of inequality \eqref{eq:condition inequality} and the failure of SP principle are not rare. Thus, careful verification is required even for commonly used discount functions.\par 
Figure \ref{figure:counter} compares the SP effect when $\lambda=0.1$ and $\lambda=1$. As is shown in the left panel, while $\lambda=0.1$, condition \eqref{eq:condition inequality} is satisfied, and hence the candidate solution obtained through the SP principle results in an equilibrium barrier. In contrast, while $\lambda = 1,$ it is clear that condition \eqref{eq:condition inequality} is not satisfied. The right panel of Figure \ref{figure:counter} demonstrates that under this scenario, that $\frac{\partial V}{\partial q}$ obtained via the SP principle could exceed the instantaneous cost $K$ even in the continuation region. This suggests that the the SP yields an unreasonable result, as the cost of an immediate action is lower than inaction in this case. This outcome intuitively explains why the smooth pasting principle fails when condition \eqref{eq:condition inequality} does not hold. 
\begin{figure}[H]
			     {    \centering
					\begin{minipage}[t]{3in}
						\centering
						
						\includegraphics[width=3in]{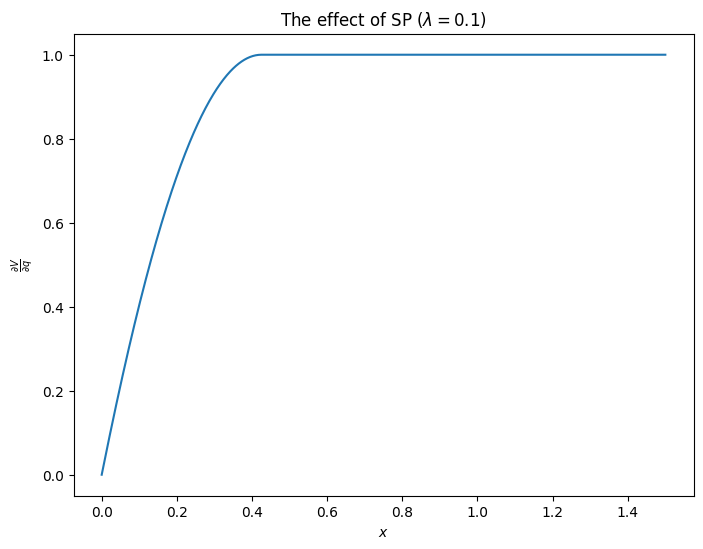}
					\end{minipage}
					\begin{minipage}[t]{3in}
						\centering
						
						\includegraphics[width=3in]{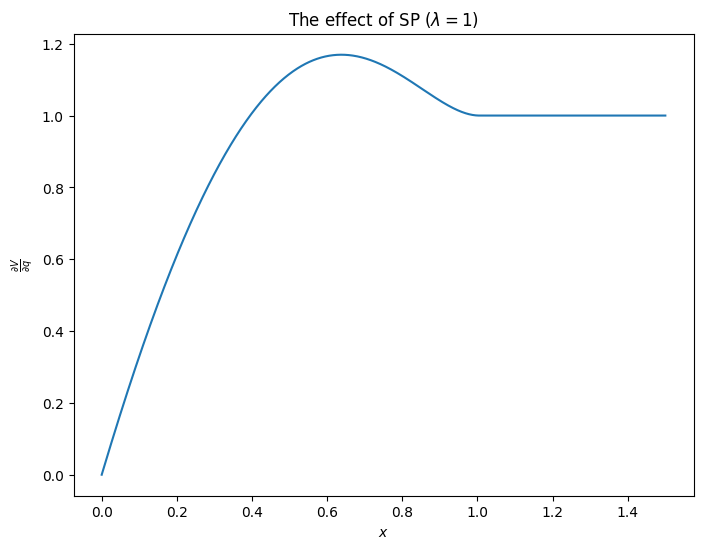}
					\end{minipage}

					\caption{The effect of smooth pasting. Parameters: $\sigma = 0.2, r = 0.05, \delta = 0.5.$}\label{figure:counter} }
               
\end{figure}\par
\section{Effects of decreasing impatience}\label{sec:decreasing impatience}
Decreasing impatience is a crucial concept in the study of behavioral time preferences. In this section, we delve into how decreasing impatience influences investment behaviors, and examine its role in the breakdown of the SP principle. To quantify decreasing impatience for a discount function, we employ the measure introduced by \cite{prelec2004decreasing}.\par 
\begin{defin}
    A discount function $h$ exhibits decreasing impatience if \cite{prelec2004decreasing} measure of decreasing impatience
$$
P(t)=-\frac{(\ln h(t))^{\prime \prime}}{(\ln h(t))^{\prime}}
$$
is non-negative.
\end{defin}\par 
\subsection{The effect of decreasing impatience on investment behaviors}
In \cite{prelec2004decreasing}, the time preference is defined by $\rho(t)=-\frac{h^{\prime}(t)}{h(t)}$ and decreasing impatience is measured by $P(t)$. Note that $\rho(t)=\rho(0) e^{-\int_{0}^{t} P(s) d s}$. Consequently, to discern the dominant effect on investment behaviors, we hold the current time preference $\rho(0)=-h^{\prime}(0)$ constant, as it remains unaffected by decreasing impatience. Proposition \ref{prop:decreasing impatience} demonstrates the effect of decreasing impatience on investment behaviors. \par 
\begin{proposition}\label{prop:decreasing impatience}
    Suppose $h_{F}$ and $h_{G}$ are weighted discount functions with weighting distributions $F$ and $G$ respectively. Define the Prelec's measure of $h_{F}\left(h_{G}\right)$ by $P_{F}\left(P_{G}\right)$. Suppose $h_{F}^{\prime}(0)=h_{G}^{\prime}(0)$. Then $x_{F}^{*}(q) \leq x_{G}^{*}(q)$ if $P_{F}(t) \geq P_{G}(t), \forall q, t>0$, where $x_{F}^{*}(q)\left(x_{G}^{*}(q)\right)$ is the triggering boundary defined by (20) with weighted distribution $F(G)$.
\end{proposition}\par 
Proposition \ref{prop:decreasing impatience} asserts that more decreasing impatience leads to more impatient investment behaviors. This finding contrasts with the results of \cite{ebert2020weighted}, who study the investment problem by time-inconsistent stopping and predict  that decreasing impatience yields a delayed investment.\par 
The drastic difference is rooted in the different interpretations of ``investment" in terms of risk attitude. We believe it would be natural to analyze the two main models in the real options approach from the perspective of risk attitude. First, in both of the models, there is randomness in the investment timing, which is the decision variable for the investment problems. Second, Prelec's measure is the Arrow-Pratt measure of risk aversion \citep{pratt1964risk,arrow1965aspects} for random times if we interpret $\ln h(t)$ as the utility function of delays. As is shown in \cite{ebert2020weighted}, $\ln h(t)$ is decreasing and convex. This suggests that $\ln h(t)$, as a utility function of random times, is anti-symmetric to a utility function of wealth and the greater value of the Prelec's measure corresponds to more risk-taking behaviors.\par 
In our irreversible investment problem, a smaller value of the triggering boundary $x_*$ indicates that the investor is more inclined towards uncertain but potentially larger cash flows rather than accepting a fixed cost at present. This stands in stark contrast to the implications found in \cite{ebert2020weighted}, where a smaller triggering boundary suggests that the investor prefers a lump-sum profit over future uncertainty. As a result, re-examining Prelec's measure from this perspective provides a clear conclusion: decreasing impatience leads to risk-taking investment behaviors in the two main models of the real options approach.\par 
Given the distinct parameters used to characterize decreasing impatience and current time preference, the hyperbolic discount function emerges as an ideal candidate for studying this phenomenon in stochastic models. The following corollary applies Proposition \ref{prop:decreasing impatience} specifically to the case of hyperbolic discounting.\par 
\begin{corollary}\label{coro:hyperbolic decreasing impatience}
    Consider the hyperbolic discount function $h(t)=\frac{1}{(1+\alpha t)^{\frac{\beta}{\alpha}}} e^{-\phi t}, \alpha>0, \beta>0, \phi \geq$ 0. Then the triggering boundary $x^{*}(q)$ defined by \eqref{eq:x* hyperbolic} decreases with respect to the degree of decreasing impatience $\alpha$.
\end{corollary}\par 
Fig \ref{fig:DI} demonstrates the barrier boundary $x^*$ as a function of $\alpha$ in hyperbolic discount function \ref{eq:hyperbolic h}.
Since $\alpha$ characterizes the degree of decreasing impatience, the numerical experiment depicted in Figure \ref{fig:DI} suggests that a higher degree decreasing impatience will lead to early investment.
\begin{figure}[H]
\includegraphics[width=1\textwidth]{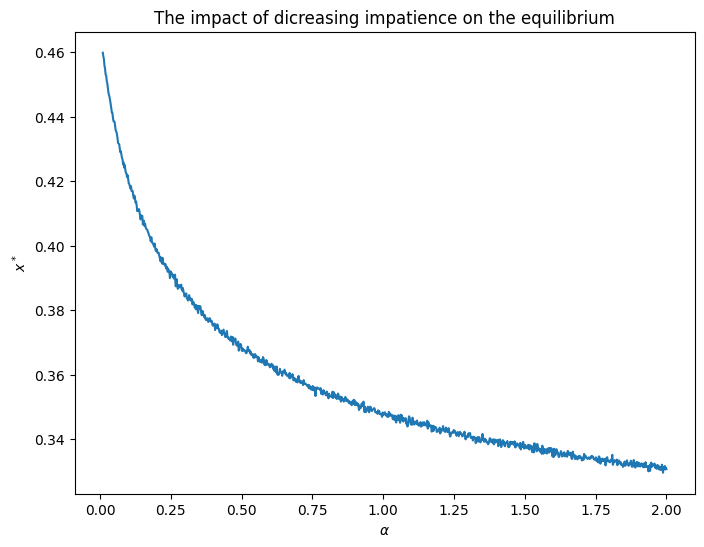}
\caption{The impact of decreasing impatience on the equilibrium. Parameters: $\sigma = 0.2, \beta = 0.05, K =1, q=1, \gamma=1.5, \phi=0.05.$}
\label{fig:DI}
\end{figure}
\par 
\subsection{Discussion: why not the (stochastic) quasi-hyperbolic discount function}
Given the widespread use of the (stochastic) quasi-hyperbolic discount function in the literature on time-inconsistent decision-making, one might wonder why we choose to move away from this simple discount function and instead focus on studying more general discount functions, such as the hyperbolic discount function, when examining the effects of decreasing impatience. We focus on this question in this subsection.\par 
 One notable drawback of this simple discount function is that none of its three parameters can characterize the decreasing impatience or the current time preference. To see this, consider the discounting mechanism of the stochastic quasi-hyperbolic discount function. Consider its discounting mechanism within the expected discounted utility framework: the present value of 1 unit of receipt at time $t$ is given by $\mathbb{E}\left[D_{0}(0, t) \times 1\right]=$ $\delta e^{-r t}+(1-\delta) e^{-(r+\lambda) t}$. Thus, the current preference is $\rho(0)=\delta r+(1-\delta)(r+\lambda)$, while decreasing impatience is captured by Prelec's measure $P(t)=\frac{a_{2}}{a_{1}}$, where $a_{1}=\frac{\delta r e^{-r t}+(1-\delta)(r+\lambda) e^{-(r+\lambda) t}}{\delta e^{-r t}+(1-\delta) e^{-(r+\lambda) t}}, a_{2}=$ $\frac{\delta r^{2} e^{-r t}+(1-\delta)(r+\lambda)^{2} e^{-(r+\lambda) t}}{\delta e^{-r t}+(1-\delta) e^{-(r+\lambda) t}}-a_{1}^{2}$. Given the intricacy of $\rho(0)$ and $P(t)$, it is challenging to analyze the effect of decreasing impatience on investment behavior using the stochastic quasi-hyperbolic discount function. This, in turn, necessitates the study of the hyperbolic discount function.

A prevalent topic in literature exploring the stochastic quasi-hyperbolic discount function is the examination of how time inconsistency affects decision-making. One approach to addressing this issue involves comparing strategies derived within an intra-game theoretic framework to a time-consistent benchmark, typically strategies obtained by exponential discounters. However, a pertinent question arises: what discount rate should be used for the time-consistent benchmark? Arguably, if we attribute time inconsistency to decreasing impatience, then the optimal choice for the benchmark discount rate would be the current time preference. This discount rate represents the rate unaffected by decreasing impatience, providing a fairer comparison to strategies derived under time inconsistency.

When considering the stochastic hyperbolic discount function \eqref{eq:stochastic quasi}, the discount rate $r$ has conventionally been utilized as the discount rate in the time-consistent benchmark \citep[e.g.,][]{grenadier2007investment}. However, $r$ does not represent the current time preference for the stochastic quasi-hyperbolic discount function when the preference is described by expected discounted payoffs. Through straightforward calculation, it is easy to see that $r=\rho(\infty)=\rho(0) e^{-\int_{0}^{\infty} P(s) d s}$. This shows that $r$ is actually the discount rate that is fully distorted by decreasing impatience. Therefore, it is inappropriate to select $r$ as the benchmark rate in the examination of decreasing impatience.

The following proposition states that different choices of benchmark discount rates will lead to contrasting economic predictions regarding investment behaviors.

\begin{proposition}\label{prop:comparison}
    Consider the stochastic quasi-hyperbolic discount function i.e., at time $t$ the agent's self $n$ applies the discount factor $D_{n}(t, s)$ given by
$$
D_{n}(t, s)=\left\{\begin{array}{rr}
e^{-r(s-t)}, & \text { if } s \in\left[t_{n}, t_{n+1}\right), \\
\delta e^{-r(s-t)}, & \text { if } s \in\left[t_{n+1}, \infty\right),
\end{array}\right.
$$
where $\left\{t_{n}\right\}_{n \geq 1}$ is a sequence of arrival times that follow a Poisson process with intensity $\lambda$ and which are independent of the shock process $X$.

Suppose $\rho(0)=\delta r+(1-\delta)(r+\lambda)$. Define the triggering boundaries obtained under the discount rate $r$ and $\rho(0)$ by respectively $x_{r}^{*}$ and $x_{\rho(0)}^{*}$, which are given by \eqref{eq:x* exponential}. Then
$$
x_{r}^{*} \leq x_{e}^{*} \leq x_{\rho(0)}^{*},
$$
where $x_{e}^{*}$ is the equilibrium investment triggering boundary given by \eqref{eq:barrier pseudo}.
\end{proposition}
\subsection{The effect of decreasing impatience on the failure of the SP principle}
\begin{proposition}\label{prop:failureSP}
    Suppose $h_{F}$ and $h_{G}$ are weighted discount functions with weighting distributions $F$ and $G$ respectively. Define the Prelec's measure of $h_{F}\left(h_{G}\right)$ by $P_{F}\left(P_{G}\right)$. Suppose $h_{F}^{\prime}(0)=h_{G}^{\prime}(0)$ and $P_{F}(t) \geq P_{G}(t), \forall q, t>0$. If inequality (\ref{eq:condition inequality}) holds for $P_{F}(t)$, then inequality (\ref{eq:condition inequality}) holds for $P_{G}(t).$
\end{proposition}
\begin{corollary}\label{coro:hyperbolic}
    Consider the hyperbolic discount function $h(t)=\frac{1}{(1+\alpha t)^{\frac{\beta}{\alpha}}} e^{-\phi t}, \alpha, \beta, \phi>0$. Suppose $\beta>\frac{\phi}{\theta(\phi)-1}$. Then
    \begin{enumerate}
        \item there exists $\alpha_{1}>0$ such that condition \eqref{eq:condition inequality} holds whenever $\alpha \in\left(0, \alpha_{1}\right)$. In this case, the equilibrium investment triggering boundary is given by \eqref{eq:x* hyperbolic}, i.e.,
\begin{equation*}
    x^{*}(q)=\frac{\int_{\phi}^{\infty} \theta(r) f\left(r-\phi ; \frac{\beta}{\alpha}, \alpha\right) d r}{\int_{\phi}^{\infty} \frac{\theta(r)-1}{r} f\left(r-\phi ; \frac{\beta}{\alpha}, \alpha\right) d r} \frac{\gamma}{\gamma-1} q^{\frac{1}{\gamma}} K .
\end{equation*}
where
$$
f(r ; k, \theta)=\frac{r^{k-1} e^{-\frac{r}{\theta}}}{\theta^{k} \Gamma(k)}, \Gamma(k)=\int_{0}^{\infty} x^{k-1} e^{-x} d x
$$

        \item there exists $\alpha_{2}>0$ such that condition (23) does not hold whenever $\alpha \in\left(\alpha_{2}, \infty\right)$.
    \end{enumerate}
\end{corollary}\par 
Proposition \ref{prop:failureSP} and Corollary \ref{coro:hyperbolic} visualized by Figure \ref{fig:vol}, reveal the reason behind the violation of inequality \eqref{eq:condition inequality}. For a hyperbolic discount function, the parameter $\alpha$ quantifies decreasing impatience and thus characterizes the degree of conflict among the time preferences of the current and future selves. Within the framework of an intra-personal game, the agent seeks the optimal solution under the assumption that future selves' strategies will be followed. In this context, the conflicted time preferences of future selves act as constraints for the current self's optimization problem. Therefore, Proposition \ref{prop:failureSP} and Corollary \ref{coro:hyperbolic} suggest that the agent may struggle to find a desirable solution if these constraints become too stringent.\par
\begin{figure}[H]
\includegraphics[width=1\textwidth]{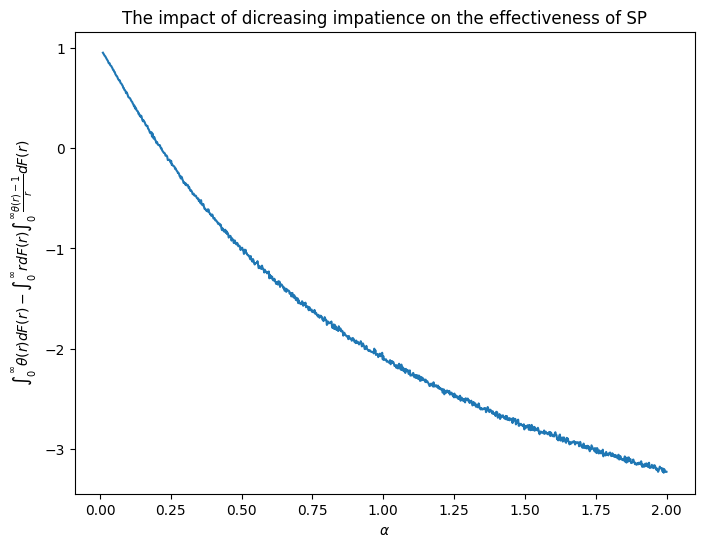}
\caption{The impact of decreasing impatience on the effectiveness of SP. Parameters: $\sigma = 0.2, \beta = 0.25.$}
\label{fig:vol}
\end{figure}\par

\section{Conclusion}\label{sec:conclusion}
This paper studies irreversible investment under time-inconsistent preferences within an intra-personal game framework. By integrating concepts from time-inconsistent control problems and weighted discount functions, we provide insights into how various factors influence investment strategies. Our analysis underscores the significance of decreasing impatience in shaping investment decisions, and cautions against blindly using the smooth pasting to solve the time-inconsistent control problem. From an economic perspective, we demonstrates that decreasing impatience prompts early investment. From a technical standpoint, we warn that decreasing impatience can lead to the failure of the smooth pasting principle.

\clearpage

\appendix

\section{Proofs}\label{appendix:proof}
For the sake of convenience, in this appendix, we define the expansion region by 
$$\mathcal{S}^{q}:=\left\{x \in(0, \infty): \frac{\partial V}{\partial q}(x, q)=K\right\} .$$ 
If $x \in \mathcal{S}^{q}$, the firm increases its output by paying $K\left(q^{\prime}-q\right)$, such that $x$ is on the boundary of $\mathcal{S}^{q^{\prime}}$. We define the non-expansion region by $\mathcal{C}^{q}$, which is the complementary set of $\mathcal{S}^{q}$, i.e., $$\mathcal{C}^{q}:=\left\{x \in(0, \infty): \frac{\partial V}{\partial q}(x, q)<K\right\}.$$

\subsection{Proof of Proposition \ref{prop:equilibrium}}
We suppose that the firm does not increase the output instantaneously at $t=0$, i.e., $x \in \overline{\mathcal{C}^{q}}$. Otherwise, if the firm makes an immediate expansion from $q$ to $q^{\prime}$, then the inequality \eqref{eq:equilibrium} becomes
$$
\liminf _{\epsilon \rightarrow 0} \frac{J^{q^{\prime}}(x ; \bar{U})-\left(J^{q^{\prime}}\left(x ; \bar{U}^{\epsilon, a}\right)-K\left(a-q^{\prime}\right)\right)}{\epsilon} \geq 0
$$
and thus the discussion boils down to the case $x \in \overline{\mathcal{C}^{q^{\prime}}}$.

We start the proof by decomposing $J^{q}\left(x ; U^{\epsilon, a}\right)$ as follows,
$$
J^{q}\left(x ; U^{\epsilon, a}\right)=A_{1}+A_{2}+A_{3}
$$
where
$$
\begin{aligned}
A_{1}&=\mathbb{E}\left[\mathbb{E}\left[\int_{\epsilon}^{\infty} h(t-\epsilon) \Pi\left(X_{t}, Q_{t}^{\bar{U}}\right) d t-\int_{\epsilon}^{\infty} K h(t-\epsilon) d Q_{t}^{\bar{U}} \mid X_{\epsilon}\right] \mid X_{0}=x\right] \\
&=\mathbb{E}\left[V\left(X_{\epsilon}, a\right) \mid X_{0}=x\right] \\
A_{2}&=\mathbb{E}\left[\int_{\epsilon}^{\infty}(h(t)-h(t-\epsilon)) \Pi\left(X_{t}, Q_{t}^{\bar{U}}\right) d t-\int_{\epsilon}^{\infty}(h(t)-h(t-\epsilon)) K d Q_{t}^{\bar{U}} \mid X_{0}=x\right] \\
A_{3}&=\mathbb{E}\left[\int_{0}^{\epsilon} h(t) \Pi\left(X_{t}, Q_{t}^{\bar{U}}\right) d t+\int_{0}^{\epsilon} h(t) K d Q_{t}^{\bar{U}} \mid X_{0}=x\right].
\end{aligned}
$$

By simple calculation, we have that
$$
\begin{aligned}
\frac{A_{1}-V(x, q)}{\epsilon} & =\mathbb{E}\left[\frac{V\left(X_{\epsilon}, a\right)-V\left(X_{\epsilon}, q\right)}{\epsilon}+\frac{V\left(X_{\epsilon}, q\right)-V(x, q)}{\epsilon} \mid X_{0}=x\right] \\
& =\mathbb{E}\left[\frac{1}{\epsilon} \int_{q}^{a}\left(\frac{\partial V}{\partial q}\left(X_{\epsilon}, s\right)-\frac{\partial V}{\partial q}(x, s)\right) d s+\frac{1}{\epsilon} \int_{q}^{a} \frac{\partial V}{\partial q}(x, s) d s\right. \\
& \left.+\frac{V\left(X_{\epsilon}, q\right)-V(x, q)}{\epsilon} \mid X_{0}=x\right]
\end{aligned}
$$

Define $\tau_{n}=\inf \left\{t>0: \frac{\partial V}{\partial q}\left(X_{t}, s\right)>n\right\}$, then it follows Ito's formula and Bellman equation \eqref{eq:HJB} that

$$
\begin{aligned}
& \mathbb{E}\left[\frac{\partial V}{\partial q}\left(X_{\epsilon \wedge \tau_{n}}, s\right)-\frac{\partial V}{\partial q}(x, s) \mid X_{0}=x\right] \\
=&\mathbb{E}\left[\int_{0}^{\epsilon \wedge \tau_{n}}\left(\frac{1}{2} \sigma^{2}(x) \frac{\partial^{2}\left(\frac{\partial V}{\partial q}\right)}{\partial x^{2}}\left(X_{t}, s\right)+\mu\left(X_{t}\right) \frac{\partial \frac{\partial V}{\partial q}}{\partial x}\left(X_{t}, s\right)\right) d s \mid X_{0}=x\right] \\
\leq & \mathbb{E}\left[\int_{0}^{\epsilon \wedge \tau_{n}}\left(\int_{0}^{\infty} r \frac{\partial w}{\partial q}\left(X_{t}, s ; r\right) d F(r)-\frac{\partial \Pi}{\partial q}\left(X_{t}, s\right) d s\right) \mid X_{0}=x\right]
\end{aligned}
$$

The growth condition of $\frac{\partial V}{\partial q}, \frac{\partial w}{\partial q}$ and $\Pi$ ensure that there exist $C>0, m>0$, such that $$
\begin{aligned}
    \left|\frac{\partial V}{\partial q}\left(X_{t}, s\right)\right|+\left|\int_{0}^{\infty} r \frac{\partial w}{\partial q}\left(X_{t}, s ; r\right) d F(r)\right|+\left|\frac{\partial \Pi}{\partial q}\left(X_{t}, s\right)\right| &\leq C\left(\left|X_{t}\right|^{m}+1\right) \\
    &\leq C\left(\sup _{0 \leq t \leq \epsilon}\left|X_{t}\right|^{m}+1\right).
\end{aligned}$$
Moreover, condition \eqref{eq:C_T} implies that $\sup _{0 \leq t \leq \epsilon}\left|X_{t}\right|^{m}$ is integrable, and thus it follows dominated convergence theorem that
$$
\begin{aligned}
&\mathbb{E}\left[\frac{\partial V}{\partial q}\left(X_{\epsilon}, s\right)-\frac{\partial V}{\partial q}(x, s) \mid X_{0}=x\right] \\
=&\lim _{n \rightarrow \infty} \mathbb{E}\left[\frac{\partial V}{\partial q}\left(X_{\epsilon \wedge \tau_{n}}, s\right)-\frac{\partial V}{\partial q}(x, s) \mid X_{0}=x\right] \\
\leq &\mathbb{E}\left[\int_{0}^{\epsilon}\left(\int_{0}^{\infty} r \frac{\partial w}{\partial q}\left(X_{t}, s ; r\right) d F(r)-\frac{\partial \Pi}{\partial q}\left(X_{t}, s\right) d s\right) \mid X_{0}=x\right]
\end{aligned}
$$
Then the continuity of $\frac{\partial w}{\partial q}$ and $\frac{\partial \Pi}{\partial q}$ yields that
\begin{equation}\label{eq:1}
\begin{aligned}
    &\limsup _{\epsilon \rightarrow 0} \frac{1}{\epsilon} \mathbb{E}\left[\frac{\partial V}{\partial q}\left(X_{\epsilon}, s\right)-\frac{\partial V}{\partial q}(x, s) \mid X_{0}=x\right] \\
    \leq &\int_{0}^{\infty} r \frac{\partial w}{\partial q}(x, s ; r) d F(r)-\frac{\partial \Pi}{\partial q}(x, s) d s
\end{aligned}
\end{equation}

Using the same argument, we have that
\begin{equation}\label{eq:2}
    \limsup _{\epsilon \rightarrow 0} \frac{1}{\epsilon} \mathbb{E}\left(V\left(X_{\epsilon}, q\right)-v(x, q)\right) \leq \int_{0}^{\infty} r w(x, q ; r) d F(r)-\Pi(x, q)
\end{equation}

By Bellman equation \eqref{eq:HJB}, we have that
\begin{equation}\label{eq:3}
    \int_{q}^{a} \frac{\partial V}{\partial q} d s \leq K(a-q)
\end{equation}

Then it follows from \eqref{eq:1}, \eqref{eq:2}, and \eqref{eq:3} that
\begin{equation}\label{eq:4}
    \limsup _{\epsilon \rightarrow 0} \frac{A_{1}-V(x, q)-K(a-q)}{\epsilon} \leq \int_{0}^{\infty} r w(x, a ; r) d F(r)-\Pi(x, a)
\end{equation}

Moreover, since $a \geq q$ and $x \in \overline{\mathcal{C}^{q}}$, condition \eqref{eq:condition uq} yields that $x \in \overline{\mathcal{C}^{a}}$. We then have
\begin{equation}\label{eq:5}
    \begin{aligned}
    &\limsup _{\epsilon \rightarrow 0} \frac{A_{2}}{\epsilon} \\ =&\mathbb{E}\left[\int_{\epsilon}^{\infty} \frac{h(t)-h(t-\epsilon)}{\epsilon} \Pi\left(X_{t}, Q_{t}^{\bar{U}}\right) d t-\int_{\epsilon}^{\infty} \frac{h(t)-h(t-\epsilon)}{\epsilon} K d Q_{t}^{\bar{U}} \mid X_{0}=x\right] \\
    =&-\int_{0}^{\infty} r w(x, a ; r) d F(r)
    \end{aligned}
\end{equation}
and
\begin{equation}\label{eq:6}
    \limsup _{\epsilon \rightarrow 0} \frac{A_{3}}{\epsilon}=\Pi(x, a)
\end{equation}

Then inequality \eqref{eq:HJB} follows from \eqref{eq:4}, \eqref{eq:5}, and \eqref{eq:6}. This completes the proof.

\subsection{Proof of Proposition \ref{prop:verification}}
We start with the sufficiency, supposing that inequality \eqref{eq:condition inequality} holds.\par 
First, it is easy to see that $\bar{U}$ determined by the triggering boundary $x^{*}(q)$ follows condition \eqref{eq:condition uq}, as $x^{*}(q)$ is increasing in $q$.\par 
Second, define 
$$\kappa(x, q):=\frac{1}{2} \sigma^{2} x^{2} \frac{\partial^{2} V}{\partial x^{2}}(x, q)+q^{1-\frac{1}{\gamma}} x-\int_{0}^{\infty} r w(x, q ; r) d F(r).$$ 
Then it follows from the derivation of the triggering boundary $x^{*}(q)$ that $\kappa(x, q)=0$ if $x<x^{*}(q)$; $\frac{\partial V}{\partial q}=K$ if $x \geq x^{*}(q)$.\par 
Third, following Proposition 3.1 in \cite{tan2021failure}, we have that $\frac{\partial V}{\partial q}(x, q)$ and $\frac{\partial w}{\partial q}(x, q)$ satisfy
\begin{equation}\label{eq:HJB verification}
    \min \left\{\frac{1}{2} \sigma^{2} x^{2} \frac{\partial^{2} \frac{\partial V}{\partial q}}{\partial x^{2}}+\left(1-\frac{1}{\gamma}\right) q^{-\frac{1}{\gamma}} x-\int_{0}^{\infty} r \frac{\partial w}{\partial q}(x, q ; r) d F(r), \frac{\partial V}{\partial q}-K\right\}=0
\end{equation}
Then it follows from \eqref{eq:HJB verification} that $\frac{\partial V}{\partial q} \leq K$.\par 
It remains to prove $\kappa(x, q) \leq 0, \forall x>0, q>0$. Since $x^{*}(q)$ is increasing and $\lim _{q \rightarrow \infty} x^{*}(q)=\infty$, we can choose, for $\forall x>0, q_{1}>q$ such that $\kappa\left(x, q_{1}\right)=0$. Note that $\kappa(x, q)=\kappa\left(x, q_{1}\right)-\int_{q}^{q_{1}} \frac{\partial \kappa}{\partial q}(x, s) d s=$ $-\int_{q}^{q_{1}} \frac{\partial \kappa}{\partial q}(x, s) d s$. Then it follows \eqref{eq:HJB verification} that $\frac{\partial \kappa}{\partial q}(x, s) \geq 0$ and thus $\kappa(x, q) \leq 0$. This completes the sufficiency.\par 
The necessity is a result of Proposition 3.1 in \cite{tan2021failure}. In fact, from the proof of Proposition 3.1 in \cite{tan2021failure}, it is easy to see that there exists $\delta>0$ such that $\frac{\partial V}{\partial q}>K$ on $\left(\delta, x^{*}(q)\right)$. This contradicts the Bellman equation \eqref{eq:HJB} and completes the proof.

\subsection{Proof of Corollary \ref{coro: stochastic quasi}}
Consider the function $$g(\lambda)=\delta \theta(r)+(1-\delta) \theta(r+\lambda)+1-(r-b+(1-\delta) \lambda)\left(\delta \frac{\theta(r)}{r}+(1-\delta) \frac{\theta(r+\lambda)}{r+\lambda}\right).$$
Corollary \ref{coro: stochastic quasi} follows from the fact that $g(0)>0, g(\infty)<0$ and the continuity of $g(\lambda)$.\par

\subsection{Proof of Proposition \ref{prop:decreasing impatience}}
\begin{lemma}\label{lemma:1}
    For $\forall r>0, C>0$, we have
$$
\frac{\sqrt{C+r}}{r}=\int_{0}^{\infty} e^{-s r} f(s ; C) d s
$$
where
\begin{equation}\label{eq:fsC}
    f(s ; C)=\frac{1}{\sqrt{\pi s}}\left(e^{-C s}+\sqrt{\pi C s} \operatorname{erf}(\sqrt{C s})\right)
\end{equation}
and erf is the error function, i.e., $\operatorname{erf}(x)=\frac{2}{\sqrt{\pi}} \int_{0}^{x} e^{-t^{2}} d t$
\end{lemma}
\begin{proof}
    It is easy to see that
$$
\frac{\sqrt{C+r}}{r}=\frac{C+r}{r \sqrt{C+r}}=\frac{C}{r \sqrt{C+r}}+\frac{1}{\sqrt{C+r}}.
$$

From the standard textbook on Laplace transform \citep[e.g.,][]{bateman1954tables}, we have
$$
\begin{aligned}
& \frac{1}{\sqrt{C+r}}=\int_{0}^{\infty} \frac{e^{-C s}}{\sqrt{\pi s}} e^{-s r} d s \\
& \frac{C}{r \sqrt{C+r}}=\sqrt{C} \int_{0}^{\infty} \operatorname{erf}(\sqrt{C s}) e^{-s r} d s,
\end{aligned}
$$
which yields the result.
\end{proof}

\begin{lemma}\label{lemma:2}
    For the function $\theta(r), r>0$ defined by \eqref{eq:theta r} and the distribution function $F$, we have that
$$
\int_{0}^{\infty} \frac{\theta(r)-1}{r} d F(r)=\int_{0}^{\infty} h_{F}(s)\left(\frac{\sqrt{2}}{\sigma} f\left(s ; \frac{1}{8} \sigma^{2}\right)-\frac{1}{2}\right) d s
$$
where $f$ is defined by \eqref{eq:fsC} and $h_{F}$ is the discount function induced by $F$.
\end{lemma}
\begin{proof}
    By the definition of $\theta(r)$, we have that
$$
\frac{\theta(r)-1}{r}=\frac{\frac{\sqrt{2}}{\sigma} \sqrt{r+\frac{1}{8} \sigma^{2}}-\frac{1}{2}}{r}.
$$

Note that $\frac{1}{r}=\int_{0}^{\infty} e^{-r s} d s$, then Lemma \ref{lemma:1} yields that
$$
\frac{\theta(r)-1}{r}=\frac{\sqrt{2}}{\sigma} \int_{0}^{\infty} e^{-s r} f\left(s ; \frac{1}{8} \sigma^{2}\right) d s-\int_{0}^{\infty} \frac{1}{2} e^{-r s} d s.
$$

Therefore, by integration by parts, we have that
$$
\begin{aligned}
\int_{0}^{\infty} \frac{\theta(r)-1}{r} d F(r) & =\int_{0}^{\infty}\left(\frac{\sqrt{2}}{\sigma} \int_{0}^{\infty} e^{-s r} f\left(s ; \frac{1}{8} \sigma^{2}\right) d s-\int_{0}^{\infty} \frac{1}{2} e^{-r s} d s\right) d F(r) \\
& =\int_{0}^{\infty}\left(\frac{\sqrt{2}}{\sigma} \int_{0}^{\infty} e^{-s r} d F(r) f\left(s ; \frac{1}{8} \sigma^{2}\right)-\frac{1}{2} \int_{0}^{\infty} e^{-r s} d F(r)\right) d s \\
& =\int_{0}^{\infty}\left(\frac{\sqrt{2}}{\sigma} h_{F}(s) f\left(s ; \frac{1}{8} \sigma^{2}\right)-\frac{1}{2} h_{F}(s)\right) d s.
\end{aligned}
$$
This completes the proof.
\end{proof}

\begin{lemma}\label{lemma:3}
    For the function $\theta(r), r>0$ defined by \eqref{eq:theta r}, consider discount functions $h_{F}$ and $h_{G}$ with weighting distributions $F$ and $G$. Then
$$
\begin{gathered}
\int_{0}^{\infty} \frac{\theta(r)-1}{r} d F(r) \geq \int_{0}^{\infty} \frac{\theta(r)-1}{r} d G(r), \\
\int_{0}^{\infty} \theta(r) d F(r) \leq \int_{0}^{\infty} \theta(r) d G(r),
\end{gathered}
$$
if $P_{F}(t) \geq P_{G}(t)$ and $h_{F}^{\prime}(0)=h_{G}^{\prime}(0), \forall t \geq 0$.
\end{lemma}
\begin{proof}
    We begin with the proof of the first inequality. Suppose that $P_{F}(t) \geq P_{G}(t)$ and $h_{F}^{\prime}(0)=h_{G}^{\prime}(0), \forall t \geq 0$. It follows from Section 2 in \cite{ebert2020weighted} that $h_{F}(t) \geq h_{G}(t), \forall t \geq 0$.

Thanks to Lemma \ref{lemma:2}, it suffices to show that
$$
\frac{\sqrt{2}}{\sigma} f\left(s ; \frac{1}{8} \sigma^{2}\right)-\frac{1}{2} \geq 0,
$$
where $f$ is defined by \eqref{eq:fsC}.

By simple calculation, it is easy to see that
$$
\lim _{s \rightarrow \infty} \frac{\sqrt{2}}{\sigma} f\left(s ; \frac{1}{8} \sigma^{2}\right)-\frac{1}{2}=0.
$$

We now show that $f$ is decreasing with $s$. In fact, for any $C>0$,
$$
\frac{\partial f}{\partial s}(s ; C)=-e^{-C s} \frac{1}{2 s \sqrt{\pi s}}<0
$$
This yields the first inequality.

Following \cite{ebert2020weighted}, we have that
$$
\int_{0}^{\infty} \theta(r) d F(r)=\frac{1}{\sqrt{2 \pi} \sigma} \int_{0}^{\infty} t^{-\frac{3}{2}} e^{-\frac{\sigma^{2}}{8} t}\left(1-h_{F}(t)\right) d t+1
$$
Then the second inequality follows from $h_{F}(t) \geq h_{G}(t), \forall t \geq 0$.
\end{proof}

Proposition \ref{prop:decreasing impatience} follows from Lemma \ref{lemma:3}.

\subsection{Proof of Corollary \ref{coro:hyperbolic decreasing impatience}}
Note that $P(t)=\frac{\alpha}{1+\alpha t}$ and $-h^{\prime}(0)=\beta+\phi$. Then the result follows from the increase of $P$ with respect to $\alpha$ and Proposition \ref{prop:decreasing impatience}.

\subsection{Proof of Proposition \ref{prop:comparison}}
Through straightforward calculation, it is easy to see that $\theta(r)$ is increasing in $r$ while $\frac{\theta(r)-1}{r}$ is decreasing in $r$. Then $x_{r}^{*} \leq x_{e}^{*}$ follows from $\delta \theta(r)+(1-\delta) \theta(r+\lambda) \leq \delta \theta(r)+(1-\delta) \theta(r) \leq \theta(r)$ and $\delta \frac{\theta(r)-1}{r}+(1-\delta) \frac{\theta(r+\lambda)-1}{r+\lambda} \geq$ $\frac{\theta(r)-1}{r}$.

Moreover, $x_{e}^{*} \leq x_{\rho(0)}^{*}$ is a direct result of Proposition \ref{prop:decreasing impatience}. This completes the proof.
\subsection{Proof of Proposition \ref{prop:failureSP}}
 Suppose that $P_{F}(t) \geq P_{G}(t)$ and $h_{F}^{\prime}(0)=h_{G}^{\prime}(0), \forall t \geq 0$. It follows from Section 2 in \cite{ebert2020weighted} that $h_{F}(t) \geq h_{G}(t), \forall t \geq 0$. Then Lemma \ref{lemma:3} yields the conclusion.

\subsection{Proof of Corollary \ref{coro:hyperbolic}}
\begin{enumerate}
    \item It is easy to see that $\frac{1}{(1+\alpha t)^{\frac{\beta}{\alpha}}} e^{-\phi t} \rightarrow e^{-(\beta+\phi) t}$, as $\alpha \rightarrow 0$. Therefore,
$$
\int_{0}^{\infty} \theta(r) d F(r)-\int_{0}^{\infty} r d F(r) \int_{0}^{\infty} \frac{\theta(r)-1}{r} d F(r) \rightarrow 1
$$
as $\alpha \rightarrow 0$. Then the first statement holds.

    \item It is easy to see that $\frac{1}{(1+\alpha t)^{\frac{\beta}{\alpha}}} e^{-\phi t} \rightarrow e^{-\phi t}$, as $\alpha \rightarrow \infty$. Moreover, $\int_{0}^{\infty} r d F(r)=-h'(0)=\beta+\phi$. Therefore,
$$
\int_{0}^{\infty} \theta(r) d F(r)-\int_{0}^{\infty} r d F(r) \int_{0}^{\infty} \frac{\theta(r)-1}{r} d F(r) \rightarrow \theta(\phi)-(\beta+\phi)\left(\frac{\theta(\phi)-1}{\phi}\right),
$$
as $\alpha \rightarrow 0$. Then it follows from $\beta>\frac{\phi}{\theta(\phi)-1}$ that $\theta(\phi)-(\beta+\phi)\left(\frac{\theta(\phi)-1}{\phi}\right)<0$. This completes the proof.
\end{enumerate}

\bibliographystyle{apalike}
\bibliography{ref}

\end{document}